\documentclass[
submission
, nomarks
]{dmtcs-episciences}

\usepackage{listings}
\usepackage[english]{babel}
\usepackage{amssymb}
\usepackage{array}
\usepackage[noend]{algorithmic}
\usepackage[plain]{algorithm}
\usepackage{subfigure}
\usepackage{tabularx}
\usepackage{textcomp}
\usepackage{array}
\usepackage[T1]{fontenc}
\usepackage{amsthm}
\usepackage{amsmath}
\usepackage{tikz}
\usepackage[round]{natbib}
\usepackage{setspace}
\theoremstyle{definition}
 \newcounter{thm}
 \newcounter{ex}
 \newcounter{re}
 \newtheorem{Proposition}[thm]{Proposition}

 \newtheorem{Example}[ex]{Example}

 \newtheorem{Definition}[thm]{Definition}

\usetikzlibrary{matrix, arrows, calc, positioning, decorations.pathmorphing, shapes.geometric}
\renewcommand{\arraystretch}{1.2}

\newcolumntype{Y}{>{\centering\arraybackslash}X}

\newcolumntype{b}{Y}
\newcolumntype{s}{>{\hsize=.75\hsize}Y}
\newcolumntype{q}{>{\hsize=.5\hsize}Y}

\floatname{algorithm}{Algorithm}
\renewcommand{\algorithmicrequire}{\textbf{Input:}}
\renewcommand{\algorithmicensure}{\textbf{Output:}}
\renewcommand{\algorithmicforall}{\textbf{for each}}

\newcommand{\vertex}{\mbox{$n$}}
\newcommand{\edge}{\mbox{$m$}}
\newcommand{\width}{\mbox{$c$}}
\newcommand{\chain}{\mbox{$\widehat{c}$}}

\author{Nicolas Boria\affiliationmark{1}
  \and Gianpiero Cabodi\affiliationmark{1}
  \and Paolo Camurati\affiliationmark{1}
  \and Marco Palena\affiliationmark{1}
  \and Paolo Pasini\affiliationmark{1}
  \and Stefano Quer\affiliationmark{1}
}
\title[A Greedy Approach to Answer Reachability Queries on DAGs]{A Greedy Approach to Answer Reachability Queries on DAGs}

\affiliation{
  Dip.\ di Automatica ed Informatica,  Politecnico di Torino, Torino, Italy\\
}
\keywords{graph theory; DAG; reachability; transitive closure; data structure; greedy algorithm.}
\received{yyyy-mm-dd}
\revised{yyyy-mm-dd}
\accepted{yyyy-mm-dd}

\begin{document}
\publicationdetails{VOL}{YYYY}{ISS}{NUM}{SUBM}
\maketitle

\begin{abstract}
Several modern applications involve huge graphs and require fast
answers to reachability queries.
In more than two decades since first proposals, several approaches
have been presented adopting on-line searches, hop labelling or
transitive closure compression.
Transitive closure compression techniques usually construct a
graph reachability index, for example by decomposing the graph into disjoint
chains.
As memory consumption is proportional to the number of chains, the
target of those algorithms is to decompose the graph into an optimal
number \width\ of chains.
However, commonly used techniques fail to meet general expectations,
are exceedingly complex, and their application on large graphs is
impractical.
The main contribution of this paper is a novel approach to construct
such reachability indexes.
The proposed method decomposes the graph into a sub-optimal number
\chain\ of chains by following a greedy strategy.
We show that, given a vertex topological order, such a
decomposition is obtained in
$\mathcal{O}(\chain \edge)$ time, and requires
$\mathcal{O}(\chain \vertex)$ space, with
\chain\ ~bounded by
$[\width \log(\frac{\vertex}{\width})]$.
We provide experimental evidence suggesting that, on different
categories of automatically generated benchmarks as well as on graphs
arising from the field of logic synthesis and formal verification, the
proposed method produces a number of chains very close to the optimum,
while significantly reducing computation time.
\end{abstract}




\section{Introduction}

The problem of answering reachability queries in an efficient way
plays an important role in several modern applications.
Domains in which this problem arises are as diverse as office systems,
software management, geographical navigation, Internet routing and XML
indexing.
For example, in the domain of automated synthesis and verification of
digital systems, both combinational and sequential circuits can be
modelled, at some level of abstraction, as directed acyclic graphs, and
it is often required to solve reachability queries on such
graphs in order to discover variable dependencies or to compute
variable scorings~\cite{date2013coi,spe2015}.

%
%
%


In order to answer reachability queries in constant time, a na{\"i}ve
approach would explicitly store the transitive closure of the graph,
investing $\Theta (\vertex^{2})$\footnote{%
  Given a graph $G=(V,E)$, we use $\vertex$ to denote the cardinality
  of the set of vertices $|V|$, and $\edge$ to indicate the
  cardinality of the set of edges $|E|$.
}
space and $\mathcal{O}(\vertex^{3})$ time.
Such an approach would obviously be impractical for large graphs,
which are often the case in several applications.
For that reason, several more sophisticated approaches avoid storing
the whole transitive closure of the graph.
The majority of the existing reachability computation
approaches belong to three main categories.

The first category~\cite{Sanders2005,TriBl2007,Yildirim2010} includes
on-line searches.
Instead of materializing the transitive closure, these methods use
auxiliary labelling information for each vertex.
This information is pre-computed and used for pruning the search space.
For example, in GRAIL~\cite{Yildirim2010} each vertex is assigned
multiple interval labels where each interval is computed by a random
depth-first traversal.
The interval can help determine whether a vertex in the search space
can be immediately pruned because it never reaches the destination
vertex.
The pre-computation of the auxiliary labelling information in these
approaches is generally quite light, and the final index size is quite
compact.
Thus, these approaches can be applicable to very large graphs.
However, the query performance is not appealing, and those methods
can be easily one or two orders of magnitude slower than the ones
belonging to the other two categories.

The second category~\cite{Cohen2002,Jin2013} includes reachability
oracles, more commonly known as hop labelling.
Each vertex $v$ is labelled with two sets:
$L_{out(v)}$, which contains hops (vertices) $v$ can reach,
and
$L_{in(v)}$, which contains hops that can reach $v$.
Given those two sets for each vertex $v$, and nothing else, it is
possible to compute whether $u$ reaches $v$ by determining whether
there is at least a common hop,
$L_{out(u)} \cap L_{in(v)} = \emptyset$.
Those methods lie in between the first category (on-line searches) and
the third one (transitive closure materialization), as hop labelling
can be considered as a factorization of the binary matrix of the
transitive closure.
Thus, it should be able to deliver more compact indices than the
transitive closure and also offer fast query performance.

The third category~\cite{Jagadish1990,dualLabelling2006,Jin2008,%
ChenChen2008,ChenChen2011,vanSchaik2011}
includes transitive closure compression approaches.
This family of approaches aims to compress the transitive closure, i.e.,
to store for each vertex $v$ a compact representation of $TC(v)$,
\textit{e.g.}, all the vertices reachable from $v$.
The reachability from vertex $v$ to $u$ is computed by checking vertex
$u$ against $TC(v)$.
Representative approaches include chain compression, interval or tree
compression, dual-labelling, path-tree, and bit-vector compression.
Existing studies show how these approaches are the fastest in terms of
query answering since checking against transitive closure $TC(v)$ is
typically quite simple (linear scan or binary search suffices).
However, the transitive closure materialization, despite compression,
is still expensive, and the index size is often the reason these
approaches are not scalable on large graphs.
Moreover, pre-computation costs can be quite ineffective for certain
applications.  The amount of additional information associated to each node
  determines a trade-off between additional memory requirements and
  efficiency in answering reachability queries.

Jagadish~\cite{Jagadish1990} proposed an efficient data structure for
answering reachability queries on DAGs, called {\em reachability index}.
This data structure relies on a decomposition of the graph into $\chain$
disjoint chains and occupies $\Theta(\chain \vertex)$ space.
In~\cite{Jagadish1990}, the minimum number \width~of chains is proved to
be equal to the width of the graph, \textit{e.g.}, the maximum number of nodes
that are mutually unreachable.
Jagadish proposed an algorithm to decompose the graph into the optimal
number of chains that involves solving a min-flow problem.
This algorithm runs in $\mathcal{O}(\vertex^3)$ time.
The author also proposes some heuristics to compute disjoint chain
covers with a sub-optimal number of chains.
Likewise,~\cite{Felsner2003} suggest a path cover technique to solve
the lowest common ancestor problem, and~\cite{Kowaluk2008} develop a
recognition algorithm for orders of small width and graphs of small
Dilworth number based on similar ideas.

Chen et al.~\cite{ChenChen2008} improved on this result by proposing
an algorithm able to decompose a graph into the optimal number of
chains that requires only
$\mathcal{O}(\vertex^2 + \width \vertex \sqrt{\width})$ time.
Once the graph has been decomposed, labelling can be done in
$\mathcal{O}(\width \edge)$
time, as described in~\cite{ChenChen2008}.

Chen et al.~\cite{ChenChen2011} propose a technique to answer
reachability queries that relies partially on a tree decomposition of
the graph and partially on a chain decomposition constructed as shown
in~\cite{ChenChen2008}.

Unfortunately, after more than two decades since first proposals,
and a long list of worthy attempts, generally used techniques fail to
meet general expectations or a exceedingly complex.
In this paper, we follow the work
by~\cite{Jagadish1990,ChenChen2008,ChenChen2011,Teuhola1996}, and we
design an algorithm adopting a simple and fast strategy which greedily
decomposes the graph into disjoint chains.
This is done by selecting the chain that includes the maximum number
of nodes at each iteration.
The technique achieves a good trade-off between space and time
complexities, and this benefits are particularly useful when dealing
with large graphs or in specific context, like in logic synthesis and
formal verification, where only a limited amount of time can be
dedicated to the pre-processing phase.
The generated index can answer reachability queries as low as in
constant time, depending on the way the labels associated to each node
can be examined, thus satisfying time constraints of time-critical
applications.

We demonstrate that our algorithm is able to construct a compact
reachability index using
$\mathcal{O}(\chain \vertex)$
time to decompose the graph into \chain\ chains, given a topological
ordering of the nodes.
It also requires
$\mathcal{O}(\chain \vertex)$ space, with $\chain$ bounded by
$[\width \log(\frac{\vertex}{\width})]$, being \width\ the minimum
number of chains.
Even though our approach leads to a non-minimum number of chains, our
experiments suggest that its results are comparable to the optimum in
terms of space, while significantly reducing the computation time.
Furthermore, in addition to provide a different time/memory trade-off
with respect to methods such
as~\cite{Jagadish1990,ChenChen2008,ChenChen2011},
our technique is rather straightforward to implement.


\subsection{Related Works and Comparison}
\label{sec:related}

The remainder of this paper is organized as follows.
Section~\ref{sec:index} illustrates the reachability index data
structure proposed in~\cite{Jagadish1990}.
The algorithm we propose is described in Section~\ref{sec:descr}.
In Section~\ref{sec:analysis} we prove time and space bounds of the
proposed method.
Section~\ref{sec:exp} presents an experimental evaluation on the
algorithm over several categories of automatically generated
benchmarks as well as on graphs arising from the field of logic
synthesis and formal verification.
Finally, in Section~\ref{sec:conclusions} we provide some summarizing
remarks about the work.





\section{DAG Reachability Index}
\label{sec:index}


\subsection{Defining a Reachability Index}
\label{sec:indexDef}

Let $D = (V, E)$ be a DAG\footnote{
  For a cyclic graph $G=(V,E)$, i.e., a graph containing cycles, it is
  possible to find all strongly connected components (SCCs) in linear
  time~\cite{Tarjan1972-SCCs}.
  Those components can be collapsed into representative vertices, such
  that all nodes in an SCC are equivalent to their representative as
  far as reachability is concerned.
}, we add to $V$ an artificial super-source
$s$ connected to all original sources, and an artificial super-sink
$t$ connected to all original sinks.
Note that the mutual reachability of any pair of nodes in the graph is
not affected by the introduction of such nodes. Nodes are arbitrarily
numbered from $1$ to $n$.

In~\cite{Jagadish1990}, nodes of the DAG are subdivided into a set
$\Pi$ of disjoint reachability chains from $s$ to $t$.
A reachability chain $\pi$ is defined as a sequence of nodes
$(v_1, v_2, \dots, v_k)$
such that, for every node $v_i \in \pi$, there exists a directed path
in the graph from it to any of its subsequent nodes
$v_j, \; i < j \leq k$.
Each reachability chain
$\pi \in \Pi$
is identified by a positive integer index starting from one.
In order for such chains to be disjoint, each node of the graph must
appear in exactly one of them, except for the artificial super-source
$s$ and super-sink $t$ which appear in all of them.
Each non-artificial node can thus be uniquely identified through a
couple of indices $(i, j)$ where $i$ identifies the chain and $j$
states the position of the node in such a chain.

The transitive closure of the DAG is materialized by labelling each of
its nodes, except for $s$ and $t$, with reachability information
towards the computed chains.
In particular, each node $v$ is labelled with \chain\ couples of indices
$(i, j)$: \chain\ is the number of chains, and each label $(i, j)$ states
that the highest node reachable from $v$ in chain $\pi_i$ is in position
$j$.
Note that for $i$ equal to the chain of $v$, $j$ corresponds to the
direct successor of $v$ in that chain.
This way chains are implicitly stored as labels.
We consider the labels associated to a node $v$ as organized into an
array $\lambda_{v}$ of length \chain.
The $i$-th element of such an array corresponds to the position $j$ of
the highest node reachable from $v$ in $\pi_i$.
We consider the whole set of labels as organized into an array
$\Lambda$ indexed by positions in some pre-determined nodes ordering.
The $k$-th element of $\Lambda$ is $\lambda_v$, where $v$ is the
$k$-th node in such ordering.

The couple $I=(\Pi,\Lambda)$ forms a DAG reachability index.
Storing such an index requires $\mathcal{O}(\chain \vertex)$ space.
Since the number of considered chains is typically much lower than the
number of nodes, the space required to store $I$ is consequently much
lower than the one required to explicitly store the transitive closure
of the graph.

\begin{Example}
In Figure~\ref{indexexample}, an example of a DAG reachability index
is provided.
Figure~\ref{indexexample}.a shows a simple DAG $D$.
A possible reachability index for the graph, with $3$ disjoint
reachability chains
$(\pi_1, \pi_2, \pi_3)$,
is presented in Figure~\ref{indexexample}.b.
Next to each node $v$, the couple of indices on top represents the
node identifier (\textit{e.g.}, $b$ is identified by $(2, 1)$, meaning
that, in the second chain, $b$ is the first node).
The array of labels beneath it tracks the highest node reachable from
$v$ in every chain (\textit{e.g.}, the highest node reachable from $b$
in the first chain is $g$, identified by the label $(1, 3)$).

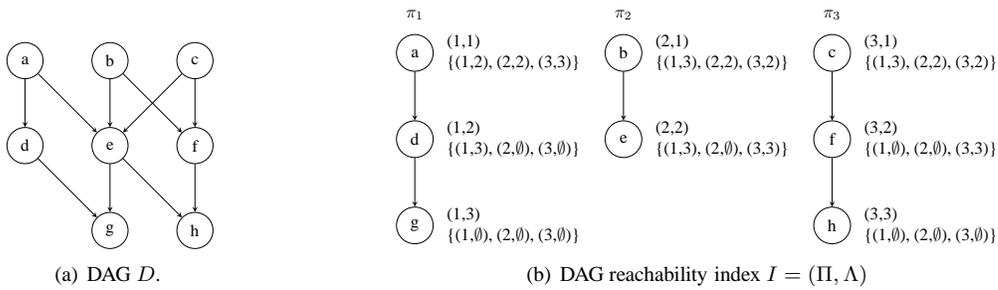
\begin{figure}[!htb]
\centering
\subfigure[DAG $D$.]
{
\resizebox{.19\linewidth}{!}{
\begin{tikzpicture}[circ/.style={circle, draw, minimum size=.75cm}, label/.style={align=left}]

\node (a) [circ] {a} ;
\node (b) [right =1cm of a, circ] {b} ;
\node (c) [right =1cm of b, circ] {c} ;
\node (d) [below =of a, circ] {d} ;
\node (e) [below =of b, circ] {e} ;
\node (f) [below =of c, circ] {f} ;
\node (g) [below =of e, circ] {g} ;
\node (h) [below =of f, circ] {h} ;

\draw[-stealth] (a) -- (d) ;
\draw[-stealth] (a) -- (e) ;
\draw[-stealth] (d) -- (g) ;
\draw[-stealth] (b) -- (e) ;
\draw[-stealth] (e) -- (g) ;
\draw[-stealth] (e) -- (h) ;
\draw[-stealth] (b) -- (f) ;
\draw[-stealth] (c) -- (e) ;
\draw[-stealth] (c) -- (f) ;
\draw[-stealth] (f) -- (h) ;

\end{tikzpicture}}
}
\quad\quad\quad\quad\quad\quad
\subfigure[DAG reachability index $I = (\Pi, \Lambda)$]
{
\resizebox{.55\linewidth}{!}{
\begin{tikzpicture}[circ/.style={circle, draw, minimum size=.75cm}, label/.style={align=left}]

\node (a) [circ] {a} ;
\node (b) [right =3.5cm of a, circ] {b} ;
\node (c) [right =3.5cm of b, circ] {c} ;
\node (d) [below =of a, circ] {d} ;
\node (e) [below =of b, circ] {e} ;
\node (f) [below =of c, circ] {f} ;
\node (g) [below =of d, circ] {g} ;
\node (h) [below =of f, circ] {h} ;

\node (la) [right= .15cm of a, label] {(1,1) \\ \{(1,2), (2,2), (3,3)\}} ;
\node (lb) [right= .15cm of b, label] {(2,1) \\ \{(1,3), (2,2), (3,2)\}} ;
\node (lc) [right= .15cm of c, label] {(3,1) \\ \{(1,3), (2,2), (3,2)\}} ;
\node (ld) [right= .15cm of d, label] {(1,2) \\ \{(1,3), (2,$\emptyset$), (3,$\emptyset$)\}} ;
\node (le) [right= .15cm of e, label] {(2,2) \\ \{(1,3), (2,$\emptyset$), (3,3)\}} ;
\node (lf) [right= .15cm of f, label] {(3,2) \\ \{(1,$\emptyset$), (2,$\emptyset$), (3,3)\}} ;
\node (lg) [right= .15cm of g, label] {(1,3) \\ \{(1,$\emptyset$), (2,$\emptyset$), (3,$\emptyset$)\}} ;
\node (lh) [right= .15cm of h, label] {(3,3) \\ \{(1,$\emptyset$), (2,$\emptyset$), (3,$\emptyset$)\}} ;

\node (p1) [above=.175cm of a] {$\pi_1$} ;
\node (p2) [above=.175cm of b] {$\pi_2$} ;
\node (p3) [above=.175cm of c] {$\pi_3$} ;

\draw[-stealth] (a) -- (d) ;
\draw[-stealth] (d) -- (g) ;
\draw[-stealth] (b) -- (e) ;
\draw[-stealth] (c) -- (f) ;
\draw[-stealth] (f) -- (h) ;

\end{tikzpicture}}
}
\caption{An example of DAG reachability index.}
\label{indexexample}
\end{figure}
\end{Example}


\subsection{Building a Reachability Index}
\label{sec:indexBuild}

In the general case, nodes of a DAG $D = (V, E)$ can be subdivided
into a variable number \chain\ of disjoint reachability chains.
An obvious upper bound for such number is \vertex, i.e., the
cardinality of the set of nodes $|V|$.
In~\cite{Jagadish1990} it is demonstrated that the width of the graph
\width\ is a tight lower bound for \chain.
As mentioned earlier, different approaches are available to compute a
set of disjoint chains $\Pi$ of either optimal or sub-optimal size.

Once $\Pi$ has been computed, nodes must be labelled.
The procedure presented in~\cite{Jagadish1990} runs in
$\mathcal{O}(\edge)$, and requires $\mathcal{O}(\chain \vertex)$
space.
The labelling procedure is independent from the number
of considered chains, and the way those chains were computed.
In order to produce such a labelling, it is required to iterate over
all outgoing edges of a node, updating an entry whenever a
reachable vertex is found in a higher position with respect to the
current state of a label, for a given path.
Overall, this step requires to check all the edges of the original
graph exactly once.

The main drawbacks of such algorithms are either the computational
effort required to find the optimum or the size of the data structures
used to store such a representation.
Furthermore, depending on the chosen algorithm, query complexity may
vary sensibly.
Last, but not least, these kind of algorithms may require a
significant effort in terms of implementation and optimization.


\subsection{Querying DAG Using a Reachability Index}
\label{sec:indexDef}

Given a DAG reachability index
$I=(\Pi,\Lambda)$,
it is possible to answer a reachability query between two nodes $v$,
$u$ in constant time.
The related procedure is described by Algorithm~\ref{alg:reachq}.

For sake of simplicity, the predicates \textsc{Chain} and
\textsc{Position} are used to retrieve some information about a node.
\textsc{Chain} returns the index of the chain a vertex belongs to,
and \textsc{Position} the position of the vertex within that chain.
If $v$ and $u$ belong to the same chain, reachability from $v$ to $u$
can be decided by checking their relative position within that
chain (line $2$).
Otherwise, if $v$ and $u$ belong to different chains, their
reachability can be determined by comparing the position of the
highest node reachable from $v$ in the chain of $u$, to the position
of $u$ itself (line $4$).

\begin{algorithm}[!htb]
\fbox{
\begin{minipage}{.97\textwidth}
\begin{spacing}{1.5}
\begin{algorithmic}[1]
\REQUIRE $I=(\Pi,\Lambda)$, $v,u \in V$
\ENSURE $\top$ or $\bot$
\IF{$\textsc{Chain}(v)=\textsc{Chain}(u)$\vspace*{-1.9mm}}
\STATE $answer\leftarrow \textsc{Position}(v)>\textsc{Position}(u)$
\ELSE \vspace*{-1.9mm}
\STATE $answer\leftarrow \textsc{Position}(\lambda_v[\textsc{Chain}(u)])>\textsc{Position}(u)$
\ENDIF
\RETURN $answer$
\end{algorithmic}
\end{spacing}
\vspace*{-1mm}
\end{minipage}
}
\caption{\textsc{ReachabilityQuery}}
\label{alg:reachq}
\end{algorithm}



\section{Algorithm Description}
\label{sec:descr}


Given a DAG $D$, and a topological order of its nodes $\Gamma$, our
algorithm produces a DAG reachability index in a simple, fast, and
efficient way.
The number of chains found is not the minimum one, but we will show that,
theoretically speaking, its upper bound is quite close to the exact
value, and, practically speaking, its error is small in many
real cases. 

To generate the reachability index, we proceed in two steps, as shown
by Algorithm~\ref{alg:highlev}.
First, we compute a disjoint chain cover
$\Pi = (\pi_1, \pi_2, ..., \pi_k)$
of $D$ (line $1$).
Then, we convert $\Pi$ into a DAG reachability index by adding the
appropriate reachability labels to each node $v \in V$ (line $2$).

\begin{algorithm}[!htb]
\fbox{
\begin{minipage}{.97\textwidth}
\begin{spacing}{1.5}
\begin{algorithmic}[1]
\REQUIRE $D=(V,E)$, $s,t \in V$, $\Gamma$
\ENSURE $\Pi, \Lambda$
\STATE $\Pi \leftarrow$ \textsc{GreedyChainCover}$(D, s, t, \Gamma)$
\STATE $\Lambda \;\leftarrow$ \textsc{NodesLabelling}$(D, t, \Pi, \Gamma)$
\RETURN $\Pi, \Lambda$
\end{algorithmic}
\end{spacing}
\vspace*{-1mm}
\end{minipage}
}
\caption{\textsc{DagReachabilityIndex}}
\label{alg:highlev}
\end{algorithm}

Before describing the algorithm in further details, we introduce the
following definitions.

\begin{Definition}
Let $D = (V, E)$ be a DAG, a chain $\pi$ is either a single node $v
\in V$ or a sequence of distinct nodes $(v_1, v_2, \dots, v_k)$,
where $v_{i+1}$ is reachable from $v_{i}$ in $D$ for each $1 \leq i \leq k-1$.
\end{Definition}

\begin{Definition}
Let $D = (V, E)$ be a DAG and let $\Pi = (\pi_1, \pi_2, ..., \pi_k)$ be a set of chains from
$s$ to $t$. A node $v \in V$ is said to be covered by $\Pi$ iff
$\exists \pi \in \Pi : v \in \pi$. Otherwise $v \in V$ is said
to be non-covered by $\Pi$.
\end{Definition}

\begin{Definition}
Let $D = (V, E)$ be a DAG and let $\Pi = (\pi_1, \pi_2, ..., \pi_k)$ be a set of chains from $s$ to
$t$. $\Pi$ is a chain cover of $D$ iff every node $v \in V$
is covered by $\Pi$. $\Pi$ is a disjoint chain cover of $D$ iff it is a chain
cover and each node belongs to exactly one $\pi \in \Pi$.
\end{Definition}

\begin{Definition}
Let $D = (V, E)$ be a DAG  and let $\Pi = (\pi_1, \pi_2, ..., \pi_k)$ be a set of chains from $s$ to
$t$. The non-covered distance metric under $\Pi$ is defined as the function $d: V
\rightarrow \mathbb{N}$ that associates to each node $v \in V$ the
maximum number of nodes non-covered by $\Pi$ in any path from $s$ to $v$.
\end{Definition}

The following two sub-sections describe procedures
\textsc{GreedyChainCover} and \textsc{NodesLabelling}.


\subsection{Greedy Chain Cover}
\label{sec:gcc}

Given $\Pi$ and $D$, in order to compute the non-covered distance of each $v \in V$ under
$\Pi$, we define a weight function $w: V \times V \rightarrow \{0, 1\}$ that assigns
a Boolean value to each edge $(v, u) \in E$ as follows:

\[
w(v, u) = \left\{
  \begin{array}{l l}
  0 & \quad\quad \text{if} \;  u \; \text{is covered by} \; \Pi\\
  1 & \quad\quad \text{otherwise}
  \end{array} \right.
\]

Each edge is initialized with unitary weight.
For each node $v \in V$, its non-covered distance under $\Pi$ can be
computed as the length of the longest path from $s$ to $v$ in the
graph $D$ weighted by $w$.
Since $D$ is a DAG, finding the longest path between two of its nodes
can be done in $\mathcal{O}(\vertex)$ time by taking into account the
topological ordering of nodes.
Computing such an ordering requires
$\mathcal{O}(\vertex + \edge)$~\cite{Lawler76}.

\begin{algorithm}[!htb]
\fbox{
\begin{minipage}{.97\textwidth}
\begin{spacing}{1.5}
\begin{algorithmic}[1]
\REQUIRE $D=(V,E)$, $s \in V$, $t \in V$, $\Gamma$
\ENSURE $\Pi$
\STATE $\Pi \leftarrow \emptyset$
\WHILE{$\exists v \in V$ non-covered} \vspace*{-1.9mm}
\STATE $\forall v \in V$ compute non-covered distance $d(v)$ following $\Gamma$
\STATE $c \leftarrow t$
\STATE $\pi \leftarrow \{t\}$
\WHILE{$s \notin \pi$} \vspace*{-1.9mm}
\STATE $p \leftarrow v \in V$ : $(v, c) \in E \wedge d(v)$ is max
\IF{$p$ is non-covered \textbf{or} $p = s$} \vspace*{-1.9mm}
\STATE \textsc{UpdateNodeInfo}$(p, \pi)$
\STATE $\pi \leftarrow \pi \cup p$
\ENDIF
\STATE $c \leftarrow p$
\ENDWHILE
\STATE $\Pi \leftarrow \Pi \cup \pi$
\ENDWHILE
\RETURN $\Pi$
\end{algorithmic}
\end{spacing}
\vspace*{-1mm}
\end{minipage}
}
\caption{\textsc{GreedyChainCover}}
\label{alg:pathcover}
\end{algorithm}

Algorithm~\ref{alg:pathcover} shows our procedure to compute a
disjoint chain cover $\Pi$ of the graph with a sub-optimal number of
chains.
$\Pi$ is constructed incrementally, by adding a new chain to a partial
cover at every iteration of the algorithm (line $12$).
Each chain is constructed backwards (lines $4$-$11$), greedily selecting
nodes based on their non-covered distance under the current $\Pi$.
First the non-covered distance for each node is computed (line $3$).
Then, starting from $t$ (lines $4$-$5$), the algorithm iteratively extends
the current chain $\pi$ with a predecessor $p$ of its last inserted
node.
This predecessor is the one whose non-covered distance $d(p)$ under
$\Pi$ is maximum (line $7$).
In order to produce a disjoint chain cover, upon selecting the next
node $p$, we check whether it is already covered by $\Pi$ or not
(line $8$).
If $p$ is not-covered, then we add it the current chain $\pi$
(line $9$) and the procedure continues.
Upon introducing a node in a chain $\pi$, its local information is
updated (line $10$).
Every node tracks the chain it belongs to, its position in $\pi$ (with
$s$ starting at zero) and its successor, if any, within $\pi$.
Otherwise such a node is not actively included in the chain $\pi$, but
it is chosen nonetheless as the next step for chain construction
(line $11$).
When $s$ is reached, $\pi$ is a chain from $s$ to $t$ with the highest
number of nodes non-covered by the current $\Pi$.
Note that the artificial nodes $s$ and $t$ are exceptions to the
disjoint character of chains as they appear in each of them.
The chain $\pi$ is then added to the current partial cover $\Pi$.
Next, the non-covered distance of every node is updated to match the
current partial chain cover and the algorithm moves to the next
iteration.
Iterations proceed until every node has been covered by $\Pi$.

\begin{Example}
Figure~\ref{figureexample} provides an example in which our greedy
approach is applied to a simple DAG.

\begin{figure}
\centering
\resizebox{.42\linewidth}{!}{
\subfigure[Graph $D$ with distances, and longest path $\pi_1$]
{
\includegraphics[scale=1]{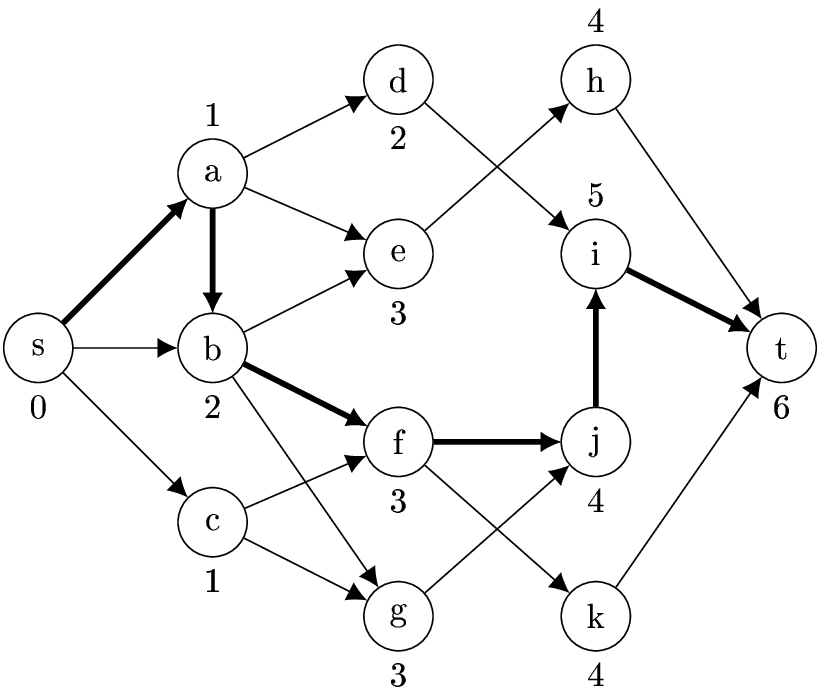}
}}
\quad\quad\quad\quad\quad\quad
\resizebox{.42\linewidth}{!}{
\subfigure[Graph $D$ at step 2 and longest path $\pi_2$]
{
\includegraphics[scale=1]{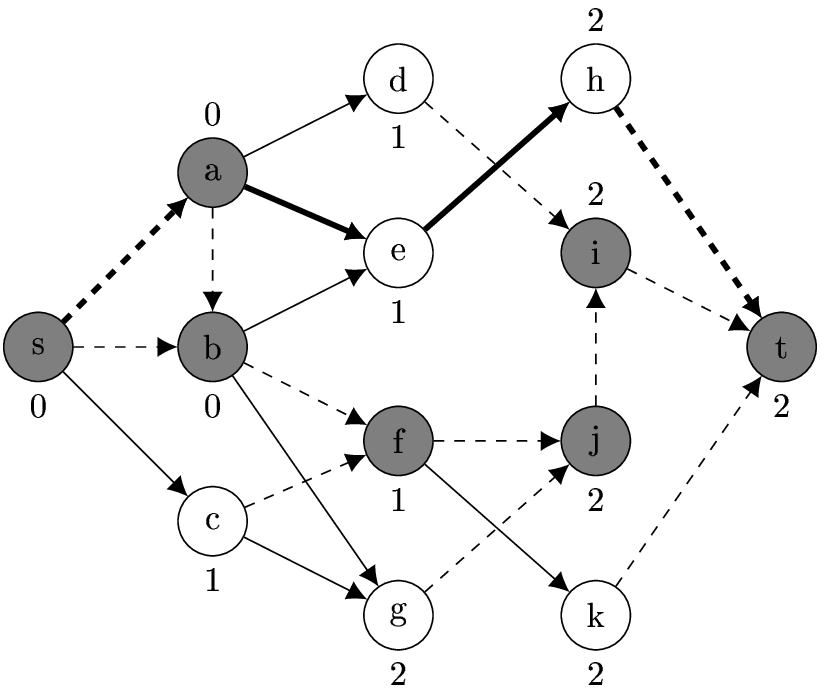}
}}
\resizebox{.42\linewidth}{!}{
\subfigure[Graph $D$ at step 3 and longest path $\pi_3$]
{
\includegraphics[scale=1]{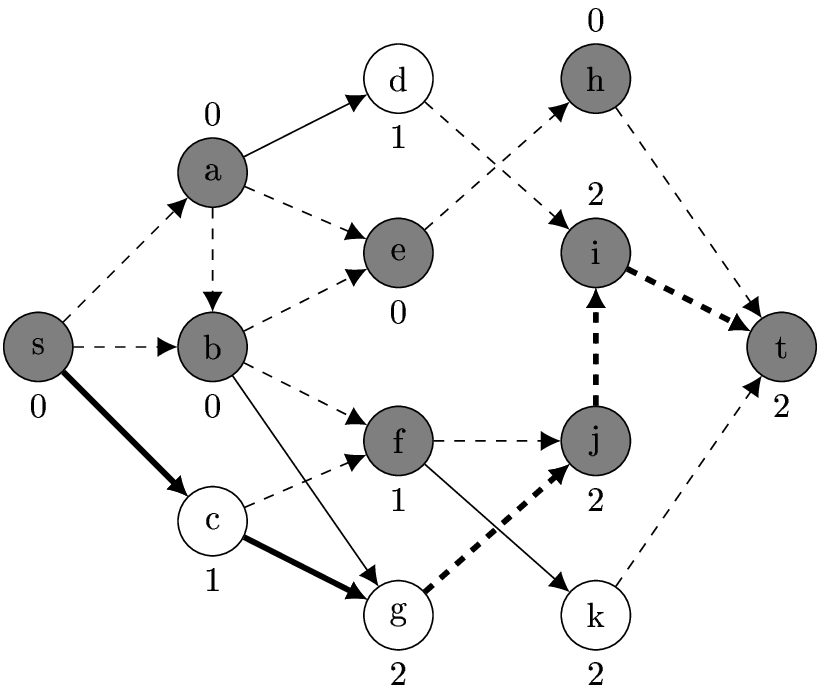}
}}
\quad\quad\quad\quad\quad\quad
\resizebox{.42\linewidth}{!}{
\subfigure[Graph $D$ at step 4 and longest path $\pi_4$]
{
\includegraphics[scale=1]{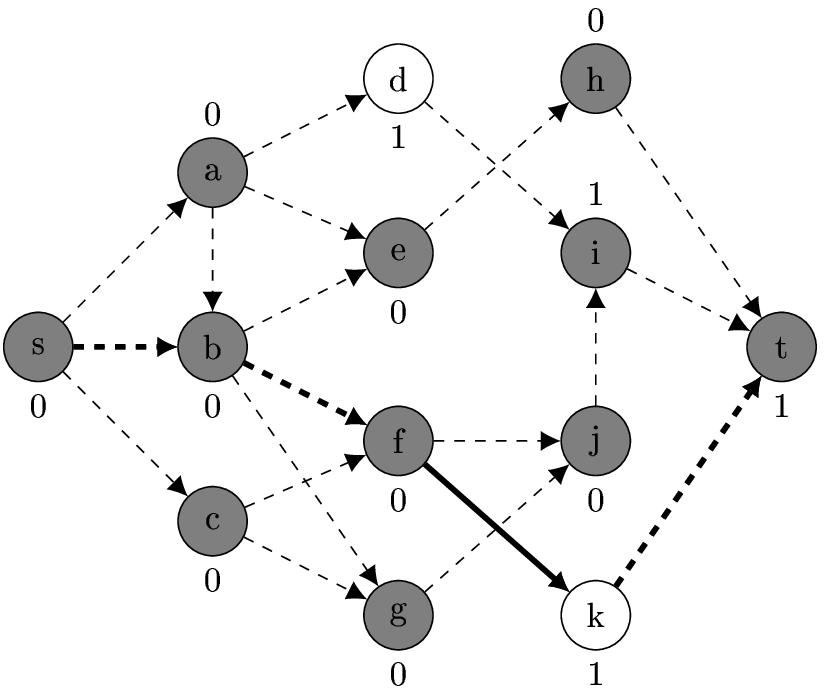}
}}
\resizebox{.42\linewidth}{!}{
\subfigure[Graph $D$ at step 5 and longest path $\pi_5$]
{
\includegraphics[scale=1]{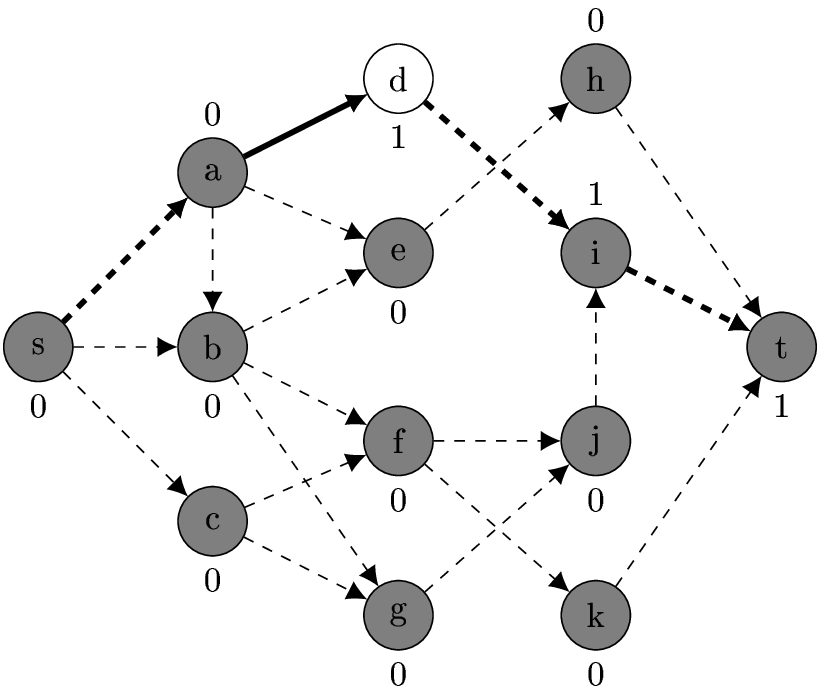}
}}
\quad
\resizebox{.54\linewidth}{!}{
\subfigure[Final set $\Pi$ of chains.]
{
\includegraphics[scale=1]{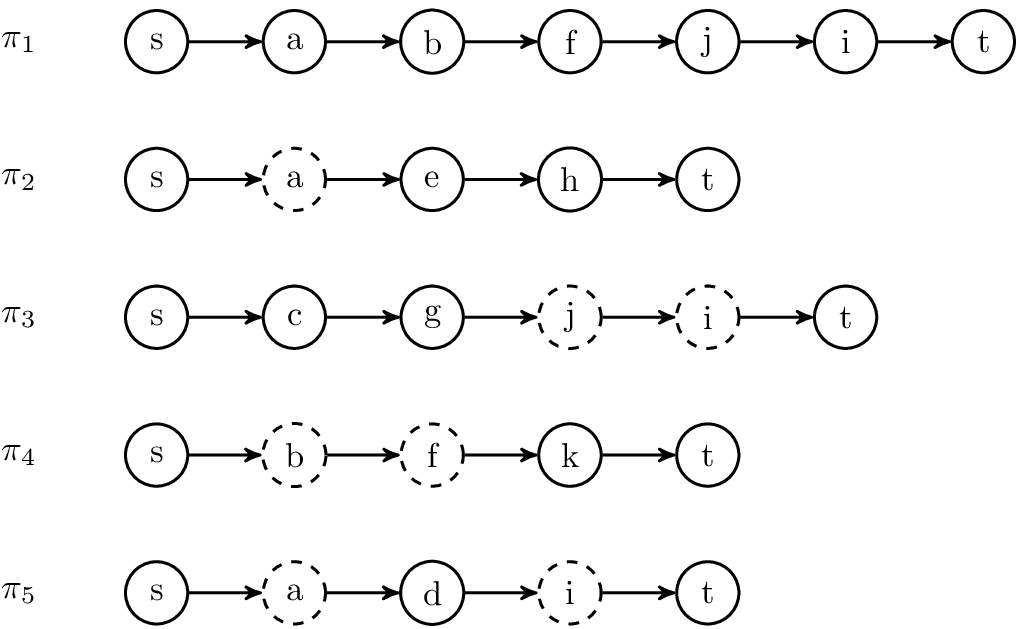}
}}
\caption{The algorithm running on a simple DAG}
\label{figureexample}
\end{figure}
At each iteration $i$, non-covered distances are updated and the chain
$\pi_i$ with longest non-covered distance from $s$ to $t$ is computed.
Edges with weight equal to zero are dashed.
The chain constructed is represented with thicker edges.
All newly covered nodes are shaded in the next sub-figure and all edges
entering such nodes are given a weight equal to zero.
Iterations continue until all nodes are covered by some chain.

Notice that the chain cover $\Pi$ constructed by our algorithm
consists of $5$ paths
$(\pi_1, \pi_2, \pi_3, \pi_4, \pi_5)$
reported in Figure~\ref{figureexample}.f, while an optimal
one would require only $4$ chains, \textit{e.g.},:
$$
\begin{array}{ll}
\pi_1: & s \rightarrow b \rightarrow f \rightarrow k \rightarrow t \\
\pi_2: & s \rightarrow c \rightarrow g \rightarrow j \rightarrow i \rightarrow t \\
\pi_3: & s \rightarrow b \rightarrow e \rightarrow h \rightarrow t \\
\pi_4: & s\rightarrow a\rightarrow d\rightarrow i\rightarrow t
\end{array}
$$
\end{Example}


\subsection{Nodes Labelling}

Algorithm~\ref{alg:vlabel} shows the labelling procedure.
In order to add cross-chain reachability information to each node we
need to examine all of its outgoing edges.
A single scan of the edges is able to address, in general, only direct
reachability.
In order to take into account also transitive reachability, we process
nodes following backward the given topological order $\Gamma$
(i.e., we adopt a reverse topological order).
In this way we ensure that when processing a node, each of its
descendants have been already labelled.
We assume that for each node $v$, its label array $\lambda_{v}$ is
initialized with \chain\ undefined labels.
Undefined labels are denoted with the symbol $\emptyset$.
We proceed by scanning each node in reverse topological order
(line $1$).
Each non-artificial node is initialized with the set of non-undefined
labels (if any) of its direct successor in its chain (line $5$), while
$s$ and $t$ are deliberately skipped (line $2$-$3$).
$\lambda_{\textsc{Next}(v)}$ is an acceptable set of labels for
$\lambda_v$, because each node reachable from $\textsc{Next}(v)$
is also reachable from $v$.

\begin{algorithm}[!htb]
\fbox{
\begin{minipage}{.97\textwidth}
\begin{spacing}{1.5}
\begin{algorithmic}[1]
\REQUIRE $D=(V,E)$, $s \in V$, $t \in V$, $\Pi$, $\Gamma$
\ENSURE $\Lambda$
\FORALL{$v \in V$ following $\Gamma$ backwards} \vspace*{-1.9mm}
\IF{$v=t$ \textbf{or} $v=s$} \vspace*{-1.9mm}
\STATE \textbf{continue}
\ENDIF
\STATE $i \leftarrow$ \textsc{Chain}$(v)$
\STATE $u \leftarrow$ \textsc{Next}$(v)$
\STATE $\lambda_{v} \leftarrow \lambda_{u}$
\STATE $\lambda_{v}[i] \leftarrow$ \textsc{Position}$(u)$
\FORALL{$w \in V : (v,w) \in E$} \vspace*{-1.9mm}
\IF{$w=t$}\vspace*{-1.9mm}
\STATE \textbf{continue}
\ENDIF
\STATE $j \leftarrow$ \textsc{Chain}$(w)$
\IF{$i \neq j$} \vspace*{-1.9mm}
\IF{$\lambda_{v}[j] = \emptyset$ {\bf or} $\textsc{Position}(w) < p$} \vspace*{-1.9mm}
\STATE $\lambda_{v}[j] \leftarrow$ \textsc{Position}$(w)$
\ENDIF
\ENDIF
\ENDFOR
\STATE $\Lambda[v] \leftarrow \lambda_{v}$
\ENDFOR
\RETURN $\Lambda$
\end{algorithmic}
\end{spacing}
\vspace*{-1mm}
\end{minipage}
}
\caption{\textsc{NodesLabelling}}
\label{alg:vlabel}
\end{algorithm}

\begin{Example}
Taking into account the graph in Figure~\ref{figureexample}, the
labelling procedure proceeds as follows. 
Nodes are processed in reverse topological order: 
$$
\begin{array}{*{13}{c}}
t & i & j & g & k & f & c & h & e & b & d & a & s
\end{array}
$$
First of all, super-sink $t$ is processed and skipped.
Then, node $i$ is processed.
Since $i$ is the non-artificial endpoint for chain $\pi_1$ no
labelling information is inherited from its successor in $\pi_1$.
Since $i$ can reach only $t$ through a direct edge, no further
labelling information is added to $\lambda_i$.
The procedure continues on node $j$.
$\lambda_j$ is initialized with $\lambda_i$ and the label for the
direct successor $i$ of $j$ in chain $\pi_1$ is set.
Since only the edge $(j, i)$ exists in the graph and both $i$ and $j$
belong to the same chain, no further labels are updated.
Node $g$ is examined next.
After the initialization of $\lambda_g$, the edge $(g, j)$ is
processed.
Since the nodes $g$ and $j$ lie on different chains, and $\lambda_g$
has no entry towards the first chain (where $j$ lies in position 4),
the entry $(1, \emptyset)$ is updated and becomes $(1, 4)$.
Node $k$ is examined next, and the algorithm behaves exactly as seen
for node $i$.
Node $f$ is examined next, and the algorithms perform the same set of
steps as seen for node $j$.

Procedure \textsc{NodesLabelling} then proceeds in a similar way for
the remaining nodes, finally producing the following labelling:
$$
\footnotesize
\begin{array}{lll}
{\bf a}: & (1, 1) & \{(1,2), (2,1), (3,2), (4,1), (5,1)\} \\
{\bf b}: & (1, 2) & \{(1,3), (2,1), (3,2), (4,1), (5,\emptyset)\} \\
{\bf c}: & (3, 1) & \{(1,3), (2,\emptyset), (3,2), (4,\emptyset), (5,\emptyset)\} \\
{\bf d}: & (5, 1) & \{(1,5), (2,\emptyset), (3,\emptyset), (4,\emptyset), (5,\emptyset)\} \\
{\bf e}: & (2, 1) & \{(1,\emptyset), (2,2), (3,\emptyset), (4,\emptyset), (5,\emptyset)\} \\
{\bf f}: & (1, 3) & \{(1,4), (2,\emptyset), (3,\emptyset), (4,1), (5,\emptyset)\} \\
{\bf g}: & (3, 2) & \{(1,4), (2,\emptyset), (3,\emptyset), (4,\emptyset), (5,\emptyset)\} \\
{\bf h}: & (2, 2) & \{(1,\emptyset), (2,\emptyset), (3,\emptyset), (4,\emptyset), (5,\emptyset)\} \\
{\bf i}: & (1, 5) & \{(1,\emptyset), (2,\emptyset), (3,\emptyset), (4,\emptyset), (5,\emptyset)\} \\
{\bf j}: & (1, 4) & \{(1,5), (2,\emptyset), (3,\emptyset), (4,\emptyset), (5,\emptyset)\} \\
{\bf k}: & (4, 1) & \{(1,\emptyset), (2,\emptyset), (3,\emptyset), (4,\emptyset), (5,\emptyset)\} 
\end{array}
$$
\end{Example}

\begin{Example}
The graph in Figure~\ref{figureexample} never leads the algorithms to
perform the step at line $14$ as a consequence of the second half of the
$if$ condition.
To illustrate a label update in case of an improvement in terms of
potential reachability, let us introduce another partial example,
illustrated in Figure~\ref{figureexample2}.

\begin{figure}[!htb]
\centering
\subfigure[Labelling before an update operation concerning node $a$.]
{
\resizebox{.55\linewidth}{!}{
\begin{tikzpicture}[circ/.style={circle, draw, minimum size=.75cm}, label/.style={align=left}]

\node (a) [circ] {a} ;
\node (b) [right =3.5cm of a, circ] {b} ;
\node (c) [right =3.5cm of b, circ] {c} ;
\node (d) [below =of a, circ] {d} ;
\node (e) [below =of b, circ] {e} ;
\node (f) [below =of c, circ] {f} ;

\node (la) [right= .15cm of a, label] {(1,1) \\ \{(1,2), (2,2), (3,$\emptyset$)\}} ;
\node (lb) [right= .15cm of b, label] {(2,1) \\ \{(1,$\emptyset$), (2,2), (3,$\emptyset$)\}} ;
\node (lc) [right= .15cm of c, label] {(3,1) \\ \{(1,$\emptyset$), (2,1), (3,2)\}} ;
\node (ld) [right= .15cm of d, label] {(1,2) \\ \{(1,$\emptyset$), (2,2), (3,$\emptyset$)\}} ;
\node (le) [right= .15cm of e, label] {(2,2) \\ \{(1,$\emptyset$), (2,$\emptyset$), (3,$\emptyset$)\}} ;
\node (lf) [right= .15cm of f, label] {(3,2) \\ \{(1,$\emptyset$), (2,$\emptyset$), (3,$\emptyset$)\}} ;

\node (p1) [above=.175cm of a] {$\pi_1$} ;
\node (p2) [above=.175cm of b] {$\pi_2$} ;
\node (p3) [above=.175cm of c] {$\pi_3$} ;

\draw[-stealth] (a) -- (d) ;
\draw[-stealth] (b) -- (e) ;
\draw[-stealth] (c) -- (f) ;

\end{tikzpicture}}
}
\quad\quad
\subfigure[Labelling after an update operation concerning node $a$.]
{
\resizebox{.55\linewidth}{!}{
\begin{tikzpicture}[circ/.style={circle, draw, minimum size=.75cm}, label/.style={align=left}]

\node (a) [circ] {a} ;
\node (b) [right =3.5cm of a, circ] {b} ;
\node (c) [right =3.5cm of b, circ] {c} ;
\node (d) [below =of a, circ] {d} ;
\node (e) [below =of b, circ] {e} ;
\node (f) [below =of c, circ] {f} ;

\node (la) [right= .15cm of a, label] {(1,1) \\ \{(1,2), (2,1), (3,$\emptyset$)\}} ;
\node (lb) [right= .15cm of b, label] {(2,1) \\ \{(1,$\emptyset$), (2,2), (3,$\emptyset$)\}} ;
\node (lc) [right= .15cm of c, label] {(3,1) \\ \{(1,$\emptyset$), (2,1), (3,2)\}} ;
\node (ld) [right= .15cm of d, label] {(1,2) \\ \{(1,$\emptyset$), (2,2), (3,$\emptyset$)\}} ;
\node (le) [right= .15cm of e, label] {(2,2) \\ \{(1,$\emptyset$), (2,$\emptyset$), (3,$\emptyset$)\}} ;
\node (lf) [right= .15cm of f, label] {(3,2) \\ \{(1,$\emptyset$), (2,$\emptyset$), (3,$\emptyset$)\}} ;

\node (p1) [above=.175cm of a] {$\pi_1$} ;
\node (p2) [above=.175cm of b] {$\pi_2$} ;
\node (p3) [above=.175cm of c] {$\pi_3$} ;

\draw[-stealth] (a) -- (d) ;
\draw[-stealth] (b) -- (e) ;
\draw[-stealth] (c) -- (f) ;

\end{tikzpicture}}
}
\caption{A second labelling example to illustrate solution improvements.}
\label{figureexample2}
\end{figure}
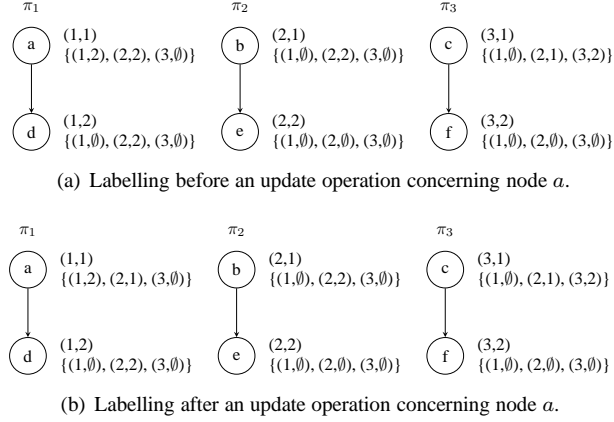

In such an example, we know that every node can reach its successor in
its own chain, and also that $d$ can reach $e$ and therefore $a$ can
reach $e$, as well as that $c$ can reach $b$.
Assume that the edge $(a, c)$ exists in the original graph, the
algorithm, whilst analyzing $a$, is going to recognize that the
highest reachable node in chain $\pi_2$ from $a$ can be updated.
Originally $a$ inherited reachability information from its successor
$d$, but the label $(2,2)$ for $a$ can be improved given the existence
of edge $(a, c)$, thus becoming  $(2, 1)$.
\end{Example}


\section{Algorithm Analysis}
\label{sec:analysis}

As mentioned earlier, the algorithm  \textsc{GreedyChainCover}
proposes a new trade-off between space and time complexities to
compute and store the reachability index structure.
More precisely, it generates a set of chains  $\Pi$ covering all nodes
of the graph in a faster way than any previously known method,
producing however a structure containing more chains, and thus larger
indexes, than these methods.

In the following we analyze both aspects of this new trade-off. First,
in Subsection \ref{approx}, we provide a simple upper bound on the
approximation ratio of the algorithm regarding the number of chains.
Then, an analysis of the complexity of our algorithm is presented in
Subsection \ref{complex}.
We also compare this complexity to those of the algorithms presented
in~\cite{Jagadish1990} and~\cite{ChenChen2008}.


\subsection{Approximation Analysis}
\label{approx}

The algorithm \textsc{GreedyChainCover} constructs a chain cover by
greedily including in the solution the chain that contains the maximal
number of non-covered vertices.
In this sense the problem of finding a chain cover on a DAG $D$ can be
re-conducted to the classical and well studied \textsc{set cover}
problem.
In this interpretation, the elements to cover are the nodes of $D$,
the sets covering the elements are all the chains from source to sink
and the goal is indeed to cover all nodes using a minimum number of
chains.

We provide an upper bound on the number of chains computed by the
proposed algorithm by proving the following result.

\begin{Proposition}
The output of \textsc{GreedyChainCover} is a set of chains of
cardinality at most
$[\width \log(\frac{\vertex}{\width})]$
that covers every node of the DAG.
\end{Proposition}

\begin{proof}
Let $D$ be a DAG with a single source $s$ and a single sink $t$ that
has \vertex\ nodes.
We wish to cover all nodes of $D$ with as few chains from $s$ to $t$
as possible.
Considering the nodes of $D$ as a base set, and the chains from $s$ to
$t$ as sets of nodes, the problem at hand is equivalent to an instance
$I$ of the \textsc{set cover} problem.
Moreover, we recall that at each of its iterations the algorithm
\textsc{GreedyChainCover} adds to the solution the chain that covers
the maximal number of nodes not yet covered by the current solution.
In this sense, \textsc{GreedyChainCover} works in the exact same way
as the classical greedy algorithm for \textsc{set cover} would on the
equivalent instance $I$.
This algorithm is known to be a
$\log(\frac{\vertex}{opt})$-approximation
algorithm, where \vertex\ is the number of elements to cover and
$\mathcal{O}pt$ the cardinality of an optimal set cover (see for
example~\cite{Vazirani2001} or~\cite{Williamson2011} and a
slight refinement of the analysis found in~\cite{ShuchiChawlaWEB}).
In our case, the value of $\mathcal{O}pt$ is known to be the width
\width\ of the graph, hence the solution computed by the greedy
algorithm is $\log(\frac{\vertex}{\width})$-approximate.
In other words its cardinality is at most
$[\width \log (\frac{\vertex}{\width})]$.
\end{proof}

In section \ref{sec:exp} we provide an experimental evaluation
regarding the quality of this approximation, considering that
$[\width \log(\frac{\vertex}{\width})]$
is merely an upper bound on the approximation ratio.
Our experimental results suggest that the number of chains produced by
\textsc{GreedyChainCover} is actually very close to the optimum.


\subsection{Complexity Analysis}
\label{complex}

In this subsection we show that \textsc{GreedyChainCover} has lower
time complexity than~\cite{ChenChen2008}.
As proved in the Subsection~\ref{approx}, the algorithm
\textsc{GreedyChainCover} computes a set of chains of cardinality at
most $[\width \log(\frac{\vertex}{\width})]$. 
This also means that the algorithm will finish after at most
$[\width \log(\frac{\vertex}{\width})]$
iterations.
At each iteration, the algorithm updates the distance of each node
in
$\mathcal{O}(\edge)$
time and then computes the current longest chain also in
$\mathcal{O}(\edge)$ time.
Hence, the overall time complexity of \textsc{GreedyChainCover} is
$\mathcal{O}(\edge \width \log(\frac{\vertex}{\width}))$.

\begin{table*}[!htb]
\begin{center}
\begin{tabular}{|c|c|c|c|c|}
\hline
 & Topological & Chains & Labelling & Space \\
 & Sort & Construction & Time & Overhead \\
\hline
\cite{Jagadish1990} &
- &
$\mathcal{O}(\vertex^3)$ &
$\mathcal{O}(\width \edge)$ &
$\mathcal{O}(\width \vertex)$\\
\cite{ChenChen2008} &
$\mathcal{O}(\vertex + \edge)$ &
$\mathcal{O}(\vertex^2 + \width \vertex \sqrt{\width})$ &
$\mathcal{O}(\width \edge)$ &
$\mathcal{O}(\width \vertex)$ \\
\textsc{GreedyChainCover} &
$\mathcal{O}(\vertex + \edge)$ &
$\mathcal{O}(\chain \edge)$ &
$\mathcal{O}(\chain \edge)$ &
$\mathcal{O}(\chain \vertex)$ \\
\hline
\end{tabular}
\vspace{2mm}
\caption{
  \label{table:comp}
   Comparison of space and time complexities.
}
\end{center}
\end{table*}

In Table~\ref{table:comp} we provide a comparison in terms of both
time and space complexities among our method and those proposed
in~\cite{Jagadish1990} and~\cite{ChenChen2008}.
Note that in our case, the labelling phase, although performed in a
similar fashion as~\cite{ChenChen2008}, involves a larger
number $\chain$ of chains with respect to $\width$.
Namely
$\mathcal{O}(\width \edge)$,
and answering a connectivity query using the labelled data structure
only requires constant time.
Of the three compared methods, only~\cite{Jagadish1990} does not
require to previously compute a topological order of nodes.
The space required to store the reachability index constructed by
\textsc{GreedyChainCover} is larger when compared to
both~\cite{ChenChen2008} and~\cite{Jagadish1990}, as it is
proportional to the number of chains.
Namely, each of the \vertex nodes of the DAG is associated to a vector
$\lambda$ of labels that has as many elements as there are chains in
the structure. Hence, instead of the
$\mathcal{O}(\width \vertex)$
space required by a data structure with an optimal number \width\ of
chains, our structure needs
$\mathcal{O}(\chain \vertex)$
space.


\section{Experimental Results}
\label{sec:exp}

In order to assess the applicability and performance of the proposed
algorithm we devised a set of experiments taking into account
different categories of graphs.

Experiments were run on an Intel Core i$7$-$3770$, with $8$ CPUs running
at $3.40$~GHz, $16$~GBytes of main memory DDR-III $1333$, and
hosting a Ubuntu $12.04$ LTS $64$ bits Linux distribution.
The time limit was set to $7200$~seconds (i.e., $2$ hours), and the
memory limit to $16$~GB.
We ran our algorithm over the selected set of benchmarks to compare
the generated number \chain\ of chains with the known optimum value
\width.
The \width\ value was computed deriving the length of the biggest
anti-chain associated with our graph when interpreted as a
{\em poset}, taking into account the theoretical aspect presented
in~\cite{Asratian1998}.
Such a procedure is quite straightforward as it simply requires to
apply a maximum matching algorithms, such as Hopcroft--Karp, on an
appropriately generated bipartite graph.
Nevertheless, for very large graph the time required to perform this
computation is non negligible, but we do not report evidence on that
issue.
Moreover, we report the time and memory require to run our algorithm.
Section~\ref{subsubsec:directComparison} also report some data to
compare our methodology with the one reported
in~\cite{ChenChen2008,ChenChen2011}.

The experiments altogether tested five groups of data:
\begin{description}
\item [Dense graphs]
In this category fall graphs with a predefined number of nodes.
Edges are added randomly until a certain density\footnote{
We define the density of a graph as $m/n^{2}$, and we usually specify
it as a percentage value.
}
threshold (either $12.5$\% or $25$\% in our case) is reached.
\item[Tree-based graphs]
This category includes graphs derived from a random tree, with a
predefined number of initial nodes, each of which with a specific
maximum degree.
Such a tree is derived from a random Pr{\"u}fer sequence.
The final graph is the created adding a given number of additional
edges.
\item[Layered graphs]
In this category fall graphs generated as a sequence of layers, each
characterized by a given number of nodes and for which each node has a
specific maximum out-degree.
An edge between two nodes is allowed only if such nodes reside in two
contiguous layers.
\item [General graphs]
To this category belong graphs generated randomly~\cite{randomGraphs},
considering one of the following criteria:
\begin{itemize}
\item
Regularity: a d-regular random graph with a delimited number of nodes.
\item
Probability: a Erd{\H o}s-R\'enyi  graph with a predefined number of
nodes, and an edge between each pair of nodes with a given
probability.
\item
Cardinality: a random graph with a bounded number of nodes and edges.
\end{itemize}
\item [Circuit-based graphs]
This category includes graphs derived from benchmarks of real-word
devices usually adopted by the formal verification community.
Given a design, we produce a benchmark by extracting the combinational
portion of the sequential circuit.
Benchmarks were selected in order to include a significant range of
variability in terms of design characteristics.
\end{description}
Notice that the first four graph sets are the ones introduced
by~\cite{ChenChen2011}.
For sake of generality, we generate those graphs with the publicly
available tool {\tt NetworkX}~\cite{networkx}.
Circuit-based benchmarks are obtained from the
\textsc{Hardware Model Checking Competition}
suites~\cite{modelCheckingCompWebPage}, and mostly come from
industrial designs.
We include this category to demonstrate the applicability of the
proposed method on practical tests from the formal verification
community, rather than just limiting ourselves to experiments on
randomly generated benchmarks.
The specific characteristics of each experimental instance, such as
number of vertices, edges, density, etc., are reported in the
following sections, alongside the results obtained.


\subsection{Comparison}
\label{subsubsec:directComparison}

As a preliminary set of experiments, we compare our technique with our
implementation of the exact one presented in~\cite{ChenChen2008}.
The original description reported on~\cite{ChenChen2008} does not
completely depict how the algorithm behaves in some limit cases
that caused the procedure to generate incorrect, that is, non-optimum,
results.
To address this issue, we took into account unreferenced material
published by the same authors.
With our final implementation, we were not able to obtain times
comparable with those proposed in~\cite{ChenChen2008}, and experiments
on graphs larger than a few thousand nodes usually did not terminate
within the slotted time limit, despite the fact that we run our
experiments on a significantly more performing hardware platform.
As a consequence, we restrict our comparison on a set of small dense
graphs such as the ones reported by~\cite{ChenChen2008}.
Table~\ref{table:smalldense} summarizes those results.

\begin{table}[!htb]
\caption{Experimental results on small dense graphs.}
\label{table:smalldense}
\begingroup
\renewcommand*{\arraystretch}{1.3}
\begin{center}
\resizebox{\linewidth}{!}{
\begin{tabularx}{.99\textwidth}{ s | s | s | s | s | s | s | s | b   }
$\vertex$ & $\edge$ & Density & Depth & Width & Chains & $\epsilon$
& \multicolumn{2}{c}{Time [s]} \\
& & [\%] & & & & [\%] & Greedy & \cite{ChenChen2008} \\
\hline
\hline
1000 &  125000 & 25 &  342 & 10 & 13 & 30.00 & 0.47 &   6.95 \\
1000 &  125000 & 25 &  339 &  9 & 11 & 22.22 & 0.50 &   6.69 \\
2000 &  500000 & 25 &  668 & 10 & 13 & 30.00 & 0.87 &  33.37 \\
2000 &  500000 & 25 &  683 &  9 & 12 & 33.33 & 0.70 &  31.22 \\
3000 & 1250000 & 25 & 1043 & 10 & 13 & 30.00 & 1.23 & 107.30 \\
3000 & 1250000 & 25 & 1056 & 11 & 12 &  9.09 & 1.32 &  94.77 \\
\hline
\end{tabularx}}
\end{center}
\endgroup
\end{table}

Columns $n$ and $m$ report the number of vertices and edges in the
graph.
For each instance, we indicate its density, depth and width
(corresponding to the minimum number of chains \width).
Not surprisingly, dense graphs present a rather low width as the high
number of connections among vertices limits the number of mutually
unreachable nodes.
As expected, on graphs with low widths our procedure yields rather
high relative error values.
Column Chains report the \chain\ value, and column $\epsilon$ the
relative error  ($((\chain - \width) / \width) \cdot 100)$).
The relative high error is due to the greedy behaviour of our
approach when running on those graphs.
The algorithm has little room to recover from sub-optimal choices made
during its first steps given the low optimum number of chains.
At the same time, tough, such a low optimum renders the impact of the
relative error rather inconsequential in practice as the impact on the
memory usage is quite moderate.
The graphs with $3000$ vertices have the same density of the
ones used by~\cite{ChenChen2008,ChenChen2011}.
In~\cite{ChenChen2008} the authors present an average labelling time
for these graphs of about $20$ seconds.
Our implementation of their algorithm is able to deal with those graphs
in about $100$ seconds, on average.
As previously noticed, our hardware platform should be faster than the
one used by the original authors, then the difference may be due to
some undefined shortcuts of our implementation or some optimizations
not disclosed by the authors.
In any case, our method perform labelling in a little more than a
second.
Then it should be at least about one order of magnitude
faster than the original algorithm presented in~\cite{ChenChen2008}.
Similar considerations hold for a direct comparison
with~\cite{ChenChen2011}, where the authors, using spanning trees
instead of chains, again on a similar set of graphs, present labelling
time in the order of $15$ seconds.


\subsection{Randomly generated graphs}
\label{subsec:random}

In this subsection we present experimental results concerning
randomly-generated graph instances.
In order to test the proposed techniques on benchmarks with different
topologies and characteristics, we take into account the four graph
categories previously introduced,
For each category, we provide a dedicated section, reporting the size
and characteristics of each benchmark instance together with the results
obtained with the proposed technique, i.e., the number of generated
chains (\chain) and the actual graph width (\width).
We also report execution time and memory for each run. 

As far as memory usage is concerned, notice that adopting the
implementation depicted in Section~\ref{sec:indexDef}, each chain
entry can be represented by an integer value.
Then, the memory usage is the one necessary to store a number of
integer values equal to $\vertex \chain$.
Chen et al.~\cite{ChenChen2008,ChenChen2011} use $16$-bit integer
values to store chains.
On the contrary, as we manipulate graphs larger than $65536$ vertices,
we use $32$-bit integer values.
Further possible optimizations, e.g., data compression
schemes such as the one proposed by van Schaik et
al.~\cite{vanSchaik2011}, are not taken into account in the tables.


\subsubsection{Dense graphs}
\label{subsubsec:dense}

Table~\ref{table:dense} illustrates the results on dense graphs larger
than the ones considered in Section~\ref{subsubsec:directComparison}.
Dense graphs tend to be characterized by low values of width: Given
the high connectivity among nodes, it is unlikely to have a large
subset of mutually unreachable nodes.
For this reason, the error can potentially grow quite large, since the
proposed greedy approach terminates in a small number of steps, and
subsequently has fewer chances to recover from sub-optimal choices.
At the same time, such an error is relative to a small number of
chains, thus its impact (in term of memory usage) is quite small in
practice.

\begin{table}[!htb]
\caption{Experimental results on dense graphs.}
\label{table:dense}
\begingroup
\renewcommand*{\arraystretch}{1.3}
\begin{center}
\resizebox{\linewidth}{!}{
\begin{tabularx}{.99\textwidth}{ s | s | s | s | s | s | s | s | s }
 $n$ & $m$ & Density & Depth & Width & Chains & $\epsilon$ & Time & Memory \\
 & & [\%] & & & & [\%] & [s] & [MB]\\
\hline
\hline
 2002 &   500008 & 12.48 &  684 & 11 & 13 & 18.18 &  0.49 & 0.10 \\
 2002 &  1000004 & 24.96 & 1146 &  6 &  7 & 16.67 &  0.87 & 0.05 \\
 4002 &  2000009 & 12.49 & 1385 & 10 & 12 & 20.00 &  1.85 & 0.18 \\
 4002 &  4000006 & 24.98 & 2278 &  7 &  7 &  0.00 &  2.44 & 0.11 \\
 6002 &  4500009 & 12.49 & 2069 & 11 & 13 & 18.18 &  3.04 & 0.31 \\
 6002 &  9000005 & 24.99 & 3446 &  7 &  7 &  0.00 &  4.16 & 0.16 \\
 8002 &  8000010 & 12.50 & 2726 & 11 & 13 & 18.18 &  4.48 & 0.42 \\
 8002 & 16000003 & 24.99 & 4573 &  7 &  7 &  0.00 &  6.99 & 0.21 \\
10002 & 12500009 & 12.50 & 3441 & 11 & 14 & 27.27 &  7.07 & 0.53 \\
10002 & 25000005 & 24.99 & 5838 &  7 &  7 &  0.00 & 11.52 & 0.27 \\
\hline
\end{tabularx}}
\end{center}
\endgroup
\end{table}

Table~\ref{table:dense} reports graphs with increasing size,
considering for each size two density values, i.e., about $12.5$\% and
$25.0$\%.
Results in terms of width, chains, errors and times are pretty
regular, then we just report results for $10$ different graphs.
All running times remain in the order of a few seconds, and memory
usage stays below $1$~MBytes for all cases.


\subsubsection{Tree-based graphs}
\label{subsubsec:tree}

In this subsection we present experimental results concerning
tree-based instances.
Such a set of graphs is characterized by a peculiar topology and
structure.
Since they are derived from a tree, these graphs tend to get wider as
the depth (in terms of distance from the root) increases.
This unusual structure is rather problematic for techniques such
as~\cite{ChenChen2008}, as it force the algorithms to introduce a
significant number of additional nodes during the chain generation
phase.
On the contrary, our proposed technique generate chains which share
many nodes in the narrower part of the graphs and then spread out
following the branches of the tree.

Table~\ref{table:tree} shows our results.

\begin{table}[!htb]
\caption{Experimental results on graphs derived from trees.}
\label{table:tree}
\begingroup
\renewcommand*{\arraystretch}{1.3}
\begin{center}
\resizebox{\linewidth}{!}{
\begin{tabularx}{.99\textwidth}{ q | q | q | q | q | q | q | q | q }
 $n$ & $m$ & Density & Depth & Width & Chains & $\epsilon$ & Time & Memory \\
 & & [\%] & & & & [\%] & [s] & [MB]\\
\hline
\hline
  5004 &  16825 & 0.067 &  767 &   670 &   753 & 11.02 &    1.52 &    14.37 \\
  5004 &  31732 & 0.127 &  937 &   319 &   393 & 18.83 &    1.08 &     7.50 \\ 
 10004 &  23635 & 0.024 & 1075 &  1801 &  1982 &  9.13 &    5.51 &    75.65\\
 10004 &  28648 & 0.029 & 1131 &  1496 &  1697 & 11.84 &    4.77 &    64.76\\
 10004 &  33607 & 0.034 & 1043 &  1255 &  1440 & 12.85 &    4.78 &    54.95\\
 15004 &  30442 & 0.014 &  986 &  3264 &  3490 &  6.48 &   16.32 &   199.75\\
 15004 &  35513 & 0.016 & 1235 &  2682 &  2973 &  9.79 &   14.18 &   170.16\\
 15004 &  40536 & 0.018 & 1319 &  2412 &  2703 & 10.77 &   11.46 &   154.71\\
 20004 &  37343 & 0.009 & 1045 &  5014 &  5273 &  4.91 &   25.29 &   402.38\\
 20004 &  42344 & 0.011 & 1242 &  4182 &  4511 &  7.29 &   24.66 &   344.23\\
 50004 &  78431 & 0.003 & 1566 & 15634 & 15970 &  2.10 &  300.61 &  3046.28\\
100004 & 151242 & 0.002 & 2339 & 19343 & 19615 &  1.40 &  911.39 &  7482.83\\
150004 & 403242 & 0.002 & 2986 & 23277 & 23435 &  0.68 & 1431.92 & 13409.97 \\
\hline
\end{tabularx}}
\end{center}
\endgroup
\end{table}

The meaning of the columns is the one described for
Table~\ref{table:smalldense}.
Errors are much smaller than for dense graphs, and they tend to decrease
for larger graphs where our greedy approach seems to close the gap with
the exact technique in terms of accuracy.
Our memory usage is then almost identical to the one required by the
exact technique.
In both cases large values of memory usage are due to the high number
of chains found for those graphs, since, as previously
mentioned, memory usage is in the order of $\mathcal{O}(\chain
\vertex)$, and it is the intrinsic limitation of transitive closure
compression techniques.
Luckily, the time required to generate the index is much smaller
than exact techniques being in the order of tens of seconds for all
benchmarks but the largest instances.


\subsubsection{Layered graphs}
\label{subsubsec:layered}

Layered graphs are characterized by a regular structure, in which
nodes are arranged in layers each presenting the same number of nodes.
Nodes in each layer can only reach nodes in the directly subsequent
layer.
This kind of structure intrinsically favours matching-based algorithms
for transitive closure compression.
The number of nodes in each layer is a trivial lower bound to the
width of the graph itself, and thus to the number of chains.
Furthermore, chains in this kind of graphs, which are usually limited
in terms of depth, tend to be short.

\begin{table}[!htb]
\caption{Experimental results on layered graphs.}
\label{table:layer}
\begingroup
\renewcommand*{\arraystretch}{1.3}
\begin{center}
\resizebox{\linewidth}{!}{
\begin{tabularx}{.99\textwidth}{  q | q | q | q | q | q | q | q | q }
 $n$ & $m$ & Density & Depth & Width & Chains & $\epsilon$ & Time & Memory \\
 & & [\%] & & & & [\%] & [s] & [MB]\\
\hline
\hline
 2502 &  12007 & 0.192 &  6 &  508 &  623 & 22.64 & 0.94 & 5.95 \\ 
 5002 &  23763 & 0.095 &  6 & 1009 & 1260 & 24.88 & 1.64 & 24.04\\
 5002 &  25517 & 0.102 & 11 &  507 &  651 & 28.40 & 1.38 & 12.42\\
 5002 &  47428 & 0.190 & 11 &  502 &  582 & 15.94 & 1.35 & 11.11\\
 5002 & 112185 & 0.448 & 11 &  501 &  535 &  6.79 & 1.99 & 10.21\\
 7502 &  38909 & 0.069 & 16 &  506 &  650 & 28.46 & 1.58 & 18.60\\
 7502 &  73523 & 0.131 & 16 &  503 &  584 & 16.10 & 2.26 & 16.71\\
 7502 &  90678 & 0.161 & 16 &  502 &  565 & 12.55 & 2.06 & 16.17\\
10002 &  29882 & 0.030 & 21 &  549 &  785 & 42.99 & 2.32 & 29.95\\
10002 &  51254 & 0.051 & 11 & 1014 & 1300 & 28.21 & 3.91 & 49.60 \\
10002 &  53093 & 0.053 & 21 &  505 &  645 & 27.72 & 2.64 & 24.61 \\
10002 &  95505 & 0.095 & 11 & 1003 & 1157 & 15.35 & 5.03 & 44.14\\
10002 & 226435 & 0.226 & 11 & 1001 & 1075 &  7.39 & 7.53 & 41.02\\
15002 &  79005 & 0.035 & 16 & 1009 & 1294 & 28.25 & 7.61 & 74.05\\
15002 & 146627 & 0.065 & 16 & 1005 & 1166 & 16.02 & 8.75 & 66.73\\
15002 & 182637 & 0.081 & 16 & 1003 & 1130 & 12.66 & 9.24 & 64.67\\
20002 &  59806 & 0.015 & 21 & 1083 & 1549 & 43.03 & 8.94 & 118.19 \\
20002 & 105982 & 0.026 & 21 & 1014 & 1303 & 28.50 & 8.93 & 99.42 \\
\hline
\end{tabularx}}
\end{center}
\endgroup
\end{table}

The experimental evaluation shows that the proposed greedy technique
produces results that are sometimes quite far from the optimum.
However, we can identify a decreasing trend of the relative error as
the number of edges increases.
Intuitively, the graph structure easily leads the greedy algorithm to
several local minimum.
This is due to the fact that a choice at a certain level is definitive
and has a direct impact on the set of reachable nodes on the layers
above.
As the number of edges increases, a sub-optimal choice at a certain
level can still be improved given the higher number of alternative
paths among nodes.

Table~\ref{table:layer} illustrates the results on layered graphs.
For time and memory figures it is possible to make similar
considerations to the one presented in Section~\ref{subsubsec:tree}.


\subsubsection{General graphs}
\label{subsubsec:general}

As a last set of randomly-generated graph, we take into account
completely random graphs created with two different generation models.
In the first model, we pre-define both the number of nodes and edges,
selecting a graph with such numbers of nodes and edges with equal
probability among all possible graphs with those characteristics.
Alternatively, in the second model, we specified just the number $n$
of vertices, and then an edge was added to the graph, between any two
given nodes, with a chosen probability.

\begin{table}[!htb]
\caption{Experimental results on ``general'' graphs with different characteristics.}
\label{table:general}
\begingroup
\renewcommand*{\arraystretch}{1.3}
\begin{center}
\resizebox{\linewidth}{!}{
\begin{tabularx}{.99\textwidth}{  q | q | q | q | q | q | q | q | q}
 $n$ & $m$ & Density & Depth & Width & Chains & $\epsilon$ & Time & Memory \\
 & & [\%] & & & & [\%] & [s] & [MB]\\
\hline
\hline
  1002 &   10102 & 1.007 &  40 &    83 &   111 & 33.73 &    0.58 &     0.42 \\
  1002 &   20051 & 1.999 &  81 &    41 &    54 & 31.71 &    0.61 &     0.21 \\
  1002 &  100011 & 9.971 & 303 &    10 &    11 & 10.00 &    0.82 &     0.04 \\
 10002 &   18659 & 0.019 &  10 &  5244 &  5490 &  4.69 &    7.71 &   209.47 \\
 10002 &   24936 & 0.025 &  15 &  3477 &  3923 & 12.83 &    7.01 &   149.68 \\
 10002 &   52634 & 0.053 &  30 &  1601 &  2026 & 26.55 &    4.79 &    77.30\\
 10002 &  100986 & 0.101 &  51 &   816 &  1085 & 32.97 &    5.34 &    41.40 \\
 10002 &  250138 & 0.250 & 115 &   334 &   445 & 33.23 &    4.34 &    16.98\\
 10002 &  501785 & 0.502 & 225 &   170 &   223 & 31.18 &    2.47 &     8.51\\
 50002 & 3123951 & 0.125 & 311 &   620 &   877 & 41.45 &   86.30 &   167.28\\
 50002 & 6248250 & 0.250 & 580 &   320 &   447 & 39.69 &   69.39 &    85.26\\
 75002 & 2342341 & 0.042 & 434 &  9254 &  9878 &  6.74 &  393.50 &  2826.19\\
100002 &  166694 & 0.002 &   9 & 41323 & 42257 &  2.26 & 1355.15 & 16120.09\\
\hline
\end{tabularx}}
\end{center}
\endgroup
\end{table}

Table~\ref{table:general} illustrates the results on general graphs.
Density is usually quite small, being those graph relatively sparse.
Exact (column Width, \width) and over-estimated chains (column Chains,
\chain) present a very high variance.
Similarly, it is possible to notice a wide range of relative errors.
The higher values are associated with graphs with lower widths,
whereas relative errors tends to decrease for larger graphs.
From the experimental results, it is possible to appreciate the
applicability of the proposed technique on graphs up to 100,000 nodes,
without an inordinate computational effort.


\subsection{Circuit-based graphs}
\label{subsec:circuit}

In this subsection we present experimental results concerning
circuit-based instances.
These kind of circuit are derived from real-life benchmarks usually
adopted in the field of hardware formal verification.
To apply Model Checking on these benchmarks it is usually required to
run a large variety of pre-processing steps with various techniques.
Many of those techniques explicitly require mutual vertex reachability,
eventually restricted to specific sub-set of
vertices~\cite{date2013coi,spe2015}, such as the ones representing
primary inputs, primary outputs, or present or next state variables.
As a consequence, this section provides some sort of case study to
illustrate the applicability of our method beyond mere randomly built
graphs.

We run our experiments on a large set of benchmarks, testing $135$
different designs derived from~\cite{modelCheckingCompWebPage}.
Figures~\ref{fig:aig}.a and~\ref{fig:aig}.b plot the error for all
designs, and, for the sake of compactness,
Tables~\ref{table:circuitDetails} and~\ref{table:circuit} report
data for just $20$ of them.
Those $20$ designs are selected as follows.
The first $10$ designs are the biggest ones in terms of number of
gates (\textit{e.g.}, graph vertices) our algorithm was able to
manipulate in the allowed time limit.
Those designs are reported in the first part of the two tables.
The last $10$ designs includes cases on which our algorithm presented
the worst performance in terms of accuracy.
Those designs are reported in the second part of the two tables.

More in details, Table~\ref{table:circuitDetails} illustrate the
characteristics of the benchmarks, i.e., the number of primary inputs
(\#PIs), latches (\#FFs), and gates (\#ANDs).
All benchmarks are single output circuits (being the output the
property under verification), thus the number of primary outputs is
omitted being always equal to $1$.
The number of primary inputs plus the number of latches determine the
number of original sources of the derived graphs.
Likewise, the number of primary outputs plus the number of latches
determine the number of original sinks.
Note that latches are modelled both as sources and sinks as they
realize the sequential behaviour of the circuits.
From a structural stand point, given the fact that the nodes within
our derived graph can either originate from latches or logic gates,
the in-degree of each node in the graph is at most two, whereas the
out degree is unbounded\footnote{
  Those benchmarks are usually stored in AIGER format, representing
  designs using AIGs (And-Inverter Graphs), i.e., using as basic
  logic blocks only NOT and two-input AND gates.
}.

\begin{table}[!htb]
\caption{Details concerning AIG-derived circuits.}
\label{table:circuitDetails}
\begingroup
\renewcommand*{\arraystretch}{1.3}
\begin{center}
\resizebox{.8\linewidth}{!}{
\begin{tabularx}{.8\textwidth}{ Y | Y | Y | Y }
Benchmark & \#PIs & \#FFs & \#ANDs \\
\hline
\hline
6s404rb1  & 202 & 9801 & 126011 \\
6s349rb06   & 254 & 14090 & 118595 \\
6s289rb05233   & 1085 & 12707 & 115953 \\
6s119   & 579 & 18833 & 107093 \\
6s290   & 543 & 13679 & 107242 \\
6s330rb06   & 1056 & 7728 & 109695 \\
6s344rb150   & 553 & 10669 & 87711 \\
6s384rb024   & 7415 & 14952 & 65415 \\
6s288r   & 12124 & 2461 & 71058 \\
6s321b5   & 22 & 13126 & 66695 \\
\hline
6s317b18   & 42 & 198 & 4849 \\
6s357r   & 17 & 196 & 1234 \\
6s268r   & 74 & 1324 & 7891 \\
6s22   & 73 & 1126  & 17243 \\
6s394r   & 17 & 180  & 930 \\
6s335rb09   & 112 & 1658 & 10813 \\
bobsmhdlc1   & 61 & 290  & 1626 \\
6s399b02   & 47 & 98  & 9783 \\
6s399b03   & 47 & 98 & 9783 \\
beemrshr2f1   & 17 & 490  & 26068 \\
\hline
\end{tabularx}}
\end{center}
\endgroup
\end{table}

Table~\ref{table:circuit} shows the results of our experiments.
All columns have the same meaning previously introduced, but the first
one, reporting the name of the benchmark.
As it can be noticed, the largest instances have more than $125,000$
vertices, and are managed in a few hundreds of seconds, obtaining
extremely small errors.

\begin{table}[!htb]
\caption{Experimental results on graphs derived from circuits.}
\label{table:circuit}
\begingroup
\renewcommand*{\arraystretch}{1.3}
\begin{center}
\resizebox{\linewidth}{!}{
\begin{tabularx}{.99\textwidth}{ b | q | q | s | q | q | q | q | q | q}
Graph & $n$ & $m$ & Density & Depth & Width & Chains & $\epsilon$ & Time & Memory \\
 & & & [\%] & & & & [\%] & [s] & [MB]\\
\hline
\hline
6s404rb1 & 136016 & 271831 & 0.001 & 137 & 26937 & 27732 & 2.95 & 831.06	& 14389.02 \\
6s289rb05233 & 129747 & 268956 & 0.002 & 28 & 42631 & 42773 & 0.33 & 1394.51 & 21170.30\\
6s119 & 126507 & 252431 & 0.002 & 111 & 37839 & 37915 & 0.20 & 927.02 & 18297.00 \\
6s290 & 121466 & 242093 & 0.002 & 43 & 42429 & 42790 & 0.85 & 1104.13 & 19827.00\\
6s330rb06 & 118481 & 236015 & 0.002 & 50 & 37697 & 37889 & 0.51 & 919.47 & 17124.66\\
6s344rb150 & 98935 & 197712 & 0.002 & 53 & 39454 & 39700 & 0.62 & 802.52 & 14983.06\\
6s384rb024 & 87784 & 174298 & 0.002 & 31 & 38252 & 38419 & 0.44 & 601.57 & 12865.35\\
6s288r & 85645 & 159161 & 0.002 & 108 & 24261 & 24313 & 0.21 & 365.26 & 7943.29\\
6s321b5 & 79845 & 159669 & 0.003 & 43 & 37231 & 37320 & 0.24 & 565.49 & 11367.09\\
\hline
6s317b18 & 5091 & 10167 & 0.039 & 45 & 1058 & 1140 & 7.75 & 1.37 & 22.14 \\
6s357r & 1449 & 2877 & 0.137 & 25 & 350 & 374 & 6.86 & 0.2 & 2.07 \\
6s268r & 9291 & 18495 & 0.021 & 35 & 2442 & 2609 & 6.84 & 4.58 & 92.47\\
6s22 & 18444 & 36811 & 0.011 & 132 & 3094 & 3305 & 6.82 & 9.69 & 232.53 \\
6s394r & 1129 & 2237 & 0.176 & 63 & 229 & 244 & 6.55 & 0.12 & 1.05\\
6s335rb09 & 12585 & 25081 & 0.016 & 43 & 4428 & 4675 & 5.58 & 8.04 & 244.44 \\
bobsmhdlc1 & 1979 & 3907 & 0.100 & 18 & 612 & 644 & 5.23 & 0.51 & 4.86\\
6s399b02 & 9930 & 19833 & 0.020 & 131 & 967 & 1012 & 4.65 & 2.29 & 38.33\\
beemrshr2f1 & 26577 & 52999 & 0.008 & 145 & 5832 & 6078 & 4.22 & 26.59 & 616.21\\
\hline
\end{tabularx}}
\end{center}
\endgroup
\end{table}

Figures~\ref{fig:aig}.a and~\ref{fig:aig}.b visually represent the
distribution of the number of chains with respect to the actual width
of the graph and the relative error distribution, respectively.
In Figure~\ref{fig:aig}.a the grey area represents the optimum,
whereas the overlying black area (almost invisible) represents the
``surplus'' due to the greedy approach.
Figure~\ref{fig:aig}.b plots the distribution of the error, sorted by
increasing values, in percentage, in order to underline how many
instances fall below a given error threshold.

\begin{figure}
\centering
\resizebox{.48\linewidth}{!}{
  \subfigure[]{\includegraphics[scale=1]{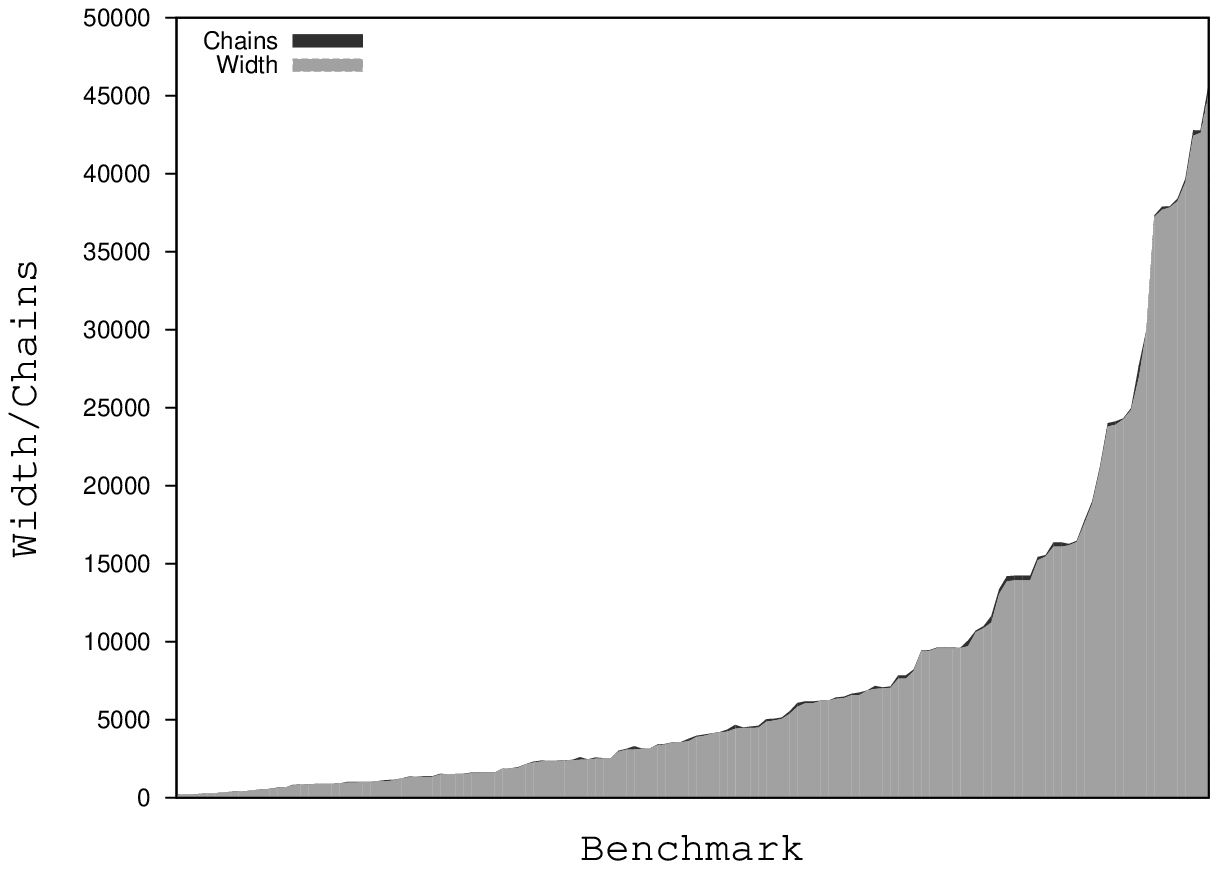}}
}
\quad
\resizebox{.48\linewidth}{!}{
  \subfigure[]{\includegraphics[scale=1]{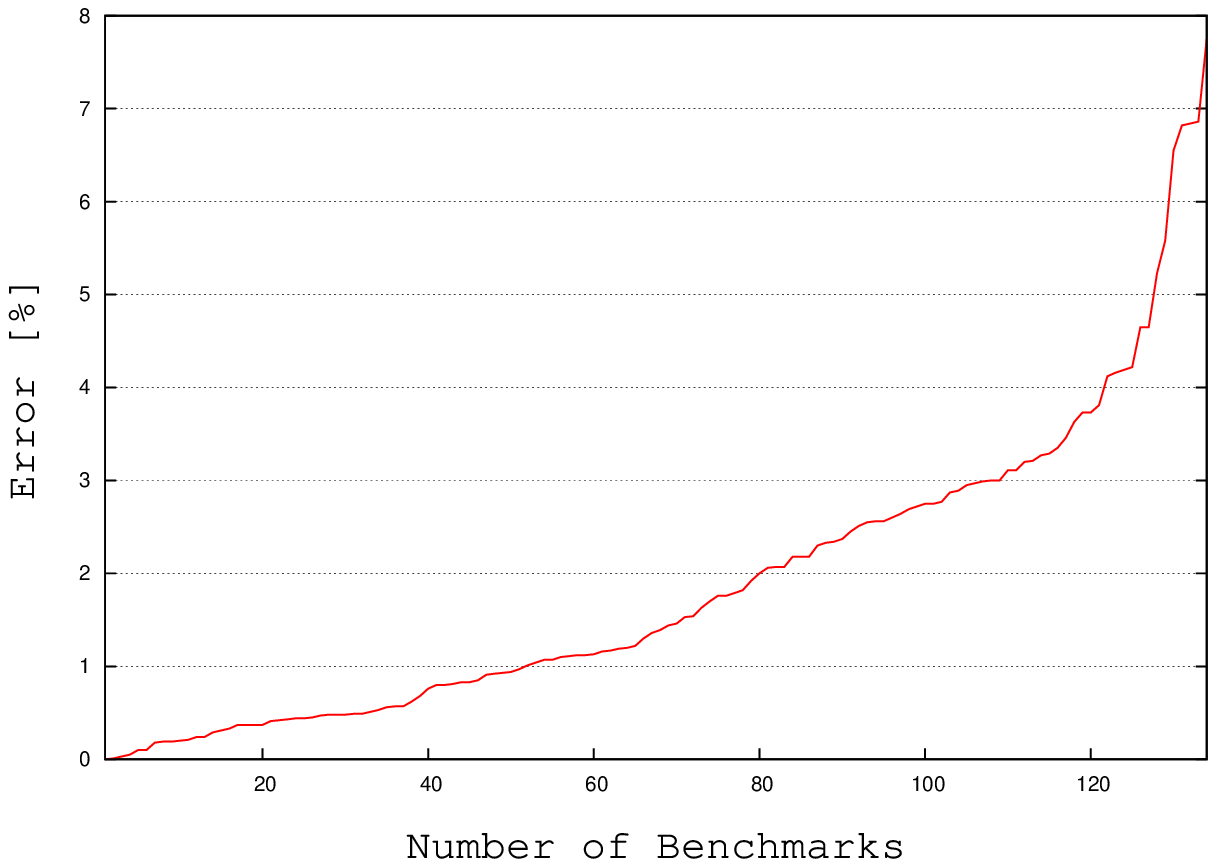}}
}
\caption{
  \label{fig:aig}
  Visual representation on experiments using circuit-derived graphs:
  Number of chains compared to exact width of the graph (a), and
  distribution of the error w.r.t. to the optimum value (b).
}
\end{figure}

The results presented shows that the number of chains found by our
greedy approach is often very close to the optimum.
The error percentage on the number of chains lies, in fact, within a
$4$\% threshold for most of the instances.
This suggests that our approach, despite being sub-optimal, can still
be applied to large graph instances without incurring in excessive
memory overhead, while significantly reducing computing time.


\section{Conclusions}
\label{sec:conclusions}

We proposed a novel greedy approach for computing the reachability
index for DAGs.
Our algorithm improves on existing methods with respect to time
requirements, while relaxing the need of optimal memory consumption
for storing such data structure.
We provided an upper bound to the time and memory complexity of our
approach.
Our experimental analysis shows that the time-memory trade-off
achieved by our algorithm is favourable, since our results are often
very close to the optimum in terms of space requirements while
significantly improving on computing time.
Moreover, the computed errors tend to decrease with graphs of
increasing size.
We deem that our approach is profitable in contexts where reachability
queries on large graphs must be answered but only a limited amount of
time can be dedicated to pre-processing.


\bibliographystyle{abbrvnat}
\bibliography{./graphTC.bbl}

\end{document}